\documentclass[journal]{IEEEtran}
\IEEEoverridecommandlockouts
\usepackage{cite}
\usepackage{amsmath,amssymb,amsfonts}
\usepackage{algorithm}
\usepackage{algorithmic}
\usepackage{graphicx}
\usepackage{textcomp}
\usepackage{xcolor}
\usepackage{autobreak}
\usepackage{amsthm}
\usepackage{changes}
\usepackage{multicol}  
\usepackage{multirow}

\newcommand{\deltav}{\boldsymbol{\delta}}
\newcommand{\phiv}{\boldsymbol{\phi}}
\newcommand{\psiv}{\boldsymbol{\psi}}
\newcommand{\av}{\mathbf{a}}
\newcommand{\bv}{\mathbf{b}}
\newcommand{\nv}{\mathbf{n}}
\newcommand{\vv}{\mathbf{v}}
\newcommand{\wv}{\mathbf{w}}
\newcommand{\xv}{\mathbf{x}}
\newcommand{\yv}{\mathbf{y}}

\newcommand{\Wv}{\mathbf{W}}
\newcommand{\gv}{\mathbf{g}}

\newcommand{\Eb}{\mathbb{E}}
\newcommand{\Cb}{\mathbb{C}}
\newcommand{\Ub}{\mathbb{U}}

\newcommand{\Real}{\mathfrak{R}}
\newcommand{\Image}{\mathfrak{I}}

\theoremstyle{definition}

\newtheorem{proposition}{Proposition}

\def\BibTeX{{\rm B\kern-.05em{\sc i\kern-.025em b}\kern-.08em
    T\kern-.1667em\lower.7ex\hbox{E}\kern-.125emX}}

\begin{document}

\title{{Max-Min Optimal Beamforming for Cell-Free Massive MIMO}}

\author{Andong Zhou$^*$, Jingxian~Wu$^*$,~\IEEEmembership{Senior~Member,~IEEE}, ~Erik G. Larsson$^{\dag}$, \IEEEmembership{Fellow, IEEE,}\\
	and Pingzhi Fan$^\ddag$, \IEEEmembership{Fellow, IEEE}
	\thanks{$^*$A. Zhou and J. Wu are with the Department of Electrical Engineering, University of Arkansas, Fayetteville, AR 72701 USA (e-mails: wuj@uark.edu, az008@email.uark.edu).}
	\thanks{$^{\dag}$E. G. Larsson is with the Department of Electrical Engineering (ISY), Link\"{o}ping University, 581 83 Link\"{o}ping, Sweden (e-mail: erik.g.larsson@liu.se).}
	\thanks{$^{\ddag}$P. Fan is with the Institute of Mobile Communications, Southwest Jiaotong University, Chengdu 611756, P. R. China (e-mail: pzfan@home.swjtu.edu.cn).}
	\thanks{The work of E. G. Larsson and P. Fan was supported in part by a STINT Sino-Swedish cooperation grant.}
}

\maketitle

\begin{abstract}
This letter develops an optimum beamforming method for downlink transmissions in cell-free massive multiple-input multiple-output (MIMO) systems, which employ a massive number of distributed access points to provide concurrent services to multiple users. The optimum design is formulated as a max-min problem that maximizes the minimum signal-to-interference-plus-noise ratio of all users. It is shown analytically that the problem is quasi-concave, and the optimum solution is obtained with the second-order cone programming. The proposed method identifies the best achievable beamforming performance in cell-free massive MIMO systems. The results can be used as benchmarks for the design of practical low complexity beamformers.
\end{abstract}

\begin{IEEEkeywords}
Cell-free massive MIMO, optimum beamforming, power control. 
\end{IEEEkeywords}

\section{Introduction}

Cell-free massive multiple-input multiple-output (MIMO) system employs a massive number of spatially distributed antennas or access points (APs) to provide concurrent services to a large number of users over a given service area \cite{ngo2015cell}, \cite{ngo2017cell}.
In cell-free massive MIMO, the coverage area is no longer segmented into cells. Instead, all APs provide services to all users in a coherent manner through beamforming. The energy efficiency and spectral efficiency of cell-free massive MIMO can be much higher than its cellular counterparts with proper beamforming designs \cite{yang2018energy, liu2019spectral}. 

The operations of both cellular and cell-free massive MIMO systems rely critically on beamforming \cite{marzetta2016fundamentals}, \cite{bjornson2019making}. When the number of antennas tends to infinity yet the number of users is fixed, it is known that the system performance is limited by pilot contamination, and eigenbeamforming  can achieve the asymptotically optimum performance for co-located massive MIMO systems \cite{ashikhmin2012pilot}.  
For cell-free massive MIMO systems, the simple and scalable conjugate beamforming (CB) technique can achieve good performance partly due to its great flexibility in the choice of power control coefficients \cite{ngo2017cell, attarifar2018random, attarifar2019modified}. 
The CB schemes can be implemented in a distributed manner without cooperation or information exchange among the APs \cite{ngo2017cell}. In \cite{nayebi2017precoding}, the performance of systems with CB or zero-forcing (ZF) are compared without considering the impacts of pilot contamination. Unlike CB, ZF is centrally designed at a central unit (CU) by considering small scale fading, thus it can outperform CB with a finite number of APs. 

The objective of this letter is to develop an optimum beamforming (OB) scheme for the downlink of cell-free massive MIMO systems with time division duplex (TDD). In TDD schemes, the APs rely on channel reciprocity to obtain the downlink channel state information (CSI) by estimating the uplink CSI with pilots from the users. The OB beamformer is centrally designed at the CU based on cooperation of all APs to balance between emphasizing desired signal and controlling multi-user interference. In order to realize the maximum potential of cell-free massive MIMO while ensuring fairness among users, the OB design is formulated as a max-min problem that aims at maximizing the minimum instantaneous signal-to-interference-plus noise ratio (SINR) among all users. It is shown through theoretical analysis that the max-min problem is quasi-concave, and the OB precoder is efficiently identified through bisection search with the help of a feasibility problem with second-order conic constraints. 
Since pilot contamination and channel estimation errors have profound impacts on the design of cell-free massive MIMO, the effects of channel estimation errors in both uplink and downlink are evaluated. The proposed OB scheme defines the best achievable beamforming performance in cell-free massive MIMO systems. Thus the results from the optimum design can be used as a benchmark for the design of low complexity beamformers such as ZF or CB.

\section{Cell-Free Massive MIMO System Model}

The cell-free massive MIMO system employs $M$ APs randomly distributed over an unbounded spatial area. The $M$ APs are connected to a CU via ideal optical front-haul links. The APs are synchronized and controlled by the CU. All $M$ APs are used to serve $K$ spatially distributed user equipment (UE) simultaneously. It is assumed that the number of APs is no less than the number of UEs, i.e., $K \leq M$. Each AP or UE is equipped with a single antenna.

The channel coefficient between the $k $th UE and the $ m $th AP can be modeled by $ g_{mk}=\sqrt{\beta_{mk}} h_{mk} \label{cm} $,
\if{\begin{equation}
g_{mk}
=\sqrt{\beta_{mk}} h_{mk} \label{cm}
\end{equation}}\fi
where $ \beta_{mk} $ and $ h_{mk} $ denote large- and small-scale fading, respectively. Large scale fading $ \beta_{mk}$ are assumed known as it changes in a much slower scale compared to the coherence interval.
The small-scale fading are independent and identically distributed (i.i.d) complex normal random variables (RVs) with zero mean and normalized to unit variance.
The channel is considered to be quasi-static, that is, the small-scale fading keeps constant within a coherence interval, and it changes into another independent random value in the next coherence interval.

The system is assumed to operate under TDD. Each time-frequency coherence channel interval is divided for uplink training, downlink training, downlink transmission, and uplink transmission. Each AP estimates the CSI for all UEs locally, and synchronize to the CU via front-haul links. Based on the estimated global CSI, the CU designs downlink beamforming precoders, and then sends the corresponding precoded information to the APs for downlink transmission. In the downlink, each UE first estimates an effective CSI that includes the effects of both the physical channels and the beamformer, and the estimation results are then used for downlink detection.

\vspace{-0.2cm}
\subsection{Uplink training}
The APs perform channel estimation with the help of uplink pilots. 
The channel from the $k$-th UE is estimated by applying the minimum mean squared error (MMSE) estimator as \cite{ngo2017cell}
\begin{equation}
\hat g_{mk}
= \frac{\sqrt{\tau_p\rho_p} \beta_{mk}} {\tau_p\rho_p \sum_{i=1}^K \beta_{mi} |\boldsymbol{\phi}^H_{k} \boldsymbol{\phi}_{i}|^2 +1} \boldsymbol{\phi}^H_{k} \yv_{m}. \label{ce}
\end{equation}
where $ \yv_{m}
=\sqrt{\tau_p \rho_p} \sum_{k=1}^K g_{mk} \boldsymbol{\phi}_k
+ \nv_m $ is the pilot signal observed by the $ m $-th AP, $\boldsymbol{\phi}_k$ is the pilot sequence from the $k$-th UE with $\|\phiv_k\|_2^2 = 1$, $\nv_m$ is additive white Gaussian noise (AWGN) vector, $\tau_c $ and $ \tau_p $ are the lengths of channel coherence interval and uplink pilot sequence, respectively, with $ \tau_c > \tau_p $, $\rho_p$ is normalized signal-to-noise ratio (SNR) of each symbol.

Since ${\hat g}_{mk}$ is a linear transformation of a Gaussian RV $\yv_m$, it is still Gaussian distributed with zero mean and variance
\begin{equation} \label{eqn:gamma}
\gamma_{mk}
= \Eb \left\{ |\hat{g}_{mk}|^2 \right\}
= \frac{\tau_p\rho_p \beta^2_{mk}} {\tau_p\rho_p \sum_{i=1}^K \beta_{mi} |\boldsymbol{\phi}^H_{k} \boldsymbol{\phi}_{i}|^2 +1}. 
\end{equation}

Denote the channel estimation error as $\epsilon_{mk} = g_{mk} - \hat{g}_{mk}$. Based on the orthogonality principle, $\epsilon_{mk}$ is uncorrelated to $\phiv_k^H \yv_m$, thus it is also uncorrelated to ${\hat g}_{mk}$. It can be easily shown that $\epsilon_{mk}$ is complex Gaussian distributed with zero mean, and the variance can be calculated as
\begin{align} \label{eqn:delta}
\delta_{mk} 
 = \Eb[( g_{mk} - \hat{g}_{mk}) g_{mk}^*] 
= \beta_{mk} - \gamma_{mk}, 
\end{align}
where the first equality is based on the orthogonality principle.

\vspace{-0.2cm}
\subsection{Downlink transmission}

Based on the channel estimation results collected from all APs, the CU designs downlink beamforming precoders to achieve uniformly good performance for all users. Denote the size $M \times K$ beamforming matrix as $\Wv = [\wv_1, \cdots, \wv_K]$, where 
$\wv_k = [w_{1k}, \cdots, w_{Mk}]^T$ is beamforming vector for the $k$-th UE. 
The designs of $\Wv$ will be discussed in next section.

After precoding, the signal transmitted by the $m$-th AP is
\begin{equation}
x_{m} =\sqrt{\rho_d} \sum_{k=1}^K w_{mk} q_k \label{ds},
\end{equation}
where $q_k$ is the symbol for the $k$-th UE with $\Eb[|q_k|^2] = 1$, and $\rho_d$ is the normalized SNR. 
The average energy of each symbol at the $m$-th AP is then
$E_m = \Eb[|x_m|^2] = \rho_d \sum_{k = 1}^K |w_{mk}|^2$.

The signal observed at the $k$-th UE can be represented by
\begin{align*}
y_{k}
 =\sqrt{\rho_d} \sum_{m=1}^M g_{mk} \sum_{i=1}^K w_{mi} q_i
+ n_k
= \sqrt{\rho_d} \sum_{i=1}^K a_{ki} q_i +n_k,
\end{align*}
where $a_{ki} = \sum_{m=1}^M g_{mk} w_{mi}$ is the effective downlink channel that includes the effects of both the physical channel and the precoder, and 
$n_k \sim \mathcal{CN}(0,1)$ is AWGN.

\vspace{-0.2cm}
\subsection{Downlink training}
Before downlink data transmission, the CU first transmits beamformed pilot sequences, such that the UE can obtain an estimate of the downlink effective channels  \cite{interdonato2016much}. Due to channel hardening achieved by the effective channel, the UE can obtain a very accurate estimate of the hardened channel gain by using pilots beamformed in the downlink \cite{interdonato2019downlink}.
The beamforming pilot signal from the $m$-th AP is
\begin{equation}
\xv_{m} = \sqrt{\tau_b} \sum_{k = 1}^K w_{mk} \psiv_k,
\end{equation}
where $\psiv_k$ is a length-$\tau_b$ pilot sequence with $\| \psiv_k\|_2 = 1$. The pilot length satisfies $ \tau_c -\tau_p > \tau_b \geq K$ such that there is no contamination. The pilot signal observed at the $k$-th UE is
\begin{align}
\yv_k = \sqrt{\tau_b \rho_b} a_{kk} \psiv_k + \sqrt{\tau_b \rho_b} \sum_{i\neq k} a_{ki} \psiv_i + \nv_k,
\end{align}
where $\rho_b$ is the normalized SNR of each downlink pilot symbol. Since it is difficult to obtain the second order statistics of the effective channel, the downlink effective channel is estimated by following the least squares criterion as
\begin{equation}
\hat{a}_{kk} = \frac{1}{\sqrt{\tau_b\rho_b}} \psiv_k^H \yv_k.
\end{equation}
The variance of the channel estimation error, $\varepsilon_k = a_{kk} - \hat{a}_{kk}$, can be evaluated as
$\Eb[|\varepsilon_k|^2] = \frac{1}{\tau_b\rho_b}$. With the estimated channel, the instantaneous SINR at the $k$-th UE is
\begin{equation}\label{SINR_UE}
\gamma_{\text{UE},k} = \frac{\rho_d|\hat{a}_{kk}|^2}{\rho_d \Eb[|\varepsilon_k|^2] + \rho_d \sum_{i \neq k}^K \Eb[|a_{ki}|^2]+1}.
\end{equation}

\section{Optimum Beamforming}

Instead of explicit AP-UE association, the optimum beamformer in cell-free massive MIMO allows each AP to serve all UEs, with the transmission power between each AP-UE link determined by the corresponding beamforming coefficient. Such an approach can achieve soft associations between AP-UE pairs.
The soft AP-UE association approach can ensure a fully coherent cooperation among all APs, thus reduces interference floor and ensures uniformly good services.

\vspace{-0.2cm}
\subsection{Problem Formulation}
We first formulate the objective and constraints of the optimum design. 
The beamforming vectors are designed by the CU. From the perspective of the CU, the signal received at the $k$-th UE can be written as 
\begin{equation}
\begin{aligned}
y_{k}
& = \underbrace{\sqrt{\rho_d} \sum_{m=1}^M \hat g_{mk} w_{mk} q_k}_{\text{DS}_k}
+ \underbrace{\sqrt{\rho_d} \sum_{m=1}^M \hat g_{mk} \sum_{i\ne k}^K w_{mi} q_i}_{\text{MUI}_k}\\
&+ \underbrace{\sqrt{\rho_d} \sum_{m=1}^M \sum_{i=1}^K \epsilon_{mk} w_{mi} q_i}_{\text{CEE}_k}
+ n_k,
\end{aligned}
\end{equation}
where the signal is divided into three parts: desired signal (DS), multi-user interference (MUI), and un-resolvable signal from channel estimation error (CEE).

With the above notation, the CU will calculate the instantaneous SINR at the $k$-th UE as
\begin{equation} \label{eqn:gammak}
 \gamma_k
 = \frac{ \left|\sum\limits_{m=1}^M \hat g_{mk} w_{mk} \right|^2 }{
 	 \sum\limits_{i\ne k}^K \left|\sum\limits_{m=1}^M \hat g_{mk} w_{mi} \right|^2 
 	+  \sum\limits_{m=1}^M \sum\limits_{i=1}^K \delta_{mk} \left|w_{mi}\right| ^2
 	+ \frac{1}{\rho_d}} .
\end{equation}
It should be noted that the SINR calculated by the CU is different from that perceived by the UE as in \eqref{SINR_UE}. Given that the CU does not have knowledge of the effective downlink channel estimated by the UE, the precoding matrix is designed at the CU by using the uplink channel estimation as in \eqref{eqn:gammak}.

To achieve the goal of offering uniformly good services for all users, we formulate the problem as a max-min problem under a per AP power constraint as
\begin{align*}
\text{(P1)~~~} \underset{\Wv}{\text{maximize}} ~~~~& \underset{k\in \{1,\dots,K\}}{\min} 
  \gamma_k
  \\
\text{subject to}\ ~~~~ &\rho_d \sum_{k = 1}^K |w_{mk}|^2 \leq \rho_d,
\ m=1,\dots, M.
\end{align*}
The power constraint in (P1) is used to ensure that the transmission power in each AP is upper bounded by $\rho_d$. The actual transmission power of the $m$-th depends on the norm of the $m$-th row of $\Wv$. The optimization of $\Wv$ will automatically optimize the transmission power of each AP. 

\subsection{Optimum Design}
The optimum solution to (P1) is identified in this subsection. Problem (P1) is non-concave due to the fact that the beamforming coefficients are on the denominator of the SINR. However, we will first show that (P1) is quasi-concave, such that efficient algorithms exist for the optimum solution of (P1).

To prove the quasi-concave property of (P1), first define a lower bound of the SINR of all UEs as
\begin{align}
\gamma_0  = \min_{ k \in \{1, \cdots, K\}} \gamma_k,
\end{align}
which implies $\gamma_k \geq \gamma_0$, for all $ k $ UEs. 
From \eqref{eqn:gammak}, the SINR constraint $\gamma_k \geq \gamma_0$ can be alternatively represented as 
\begin{equation}\label{qos}
\begin{split}
\frac{1}{\gamma_0} \left| \hat{\gv}_k^T \wv_k \right|^2
\geq \sum_{i\ne k}^K \left| \hat{\gv}_k^T \wv_i \right|^2
+ \sum_{m=1}^M \sum_{i=1}^K \delta_{mk} \left|w_{mi}\right| ^2
+ \frac{1}{\rho_d},
\end{split}
\end{equation}
where $\hat{\gv}_k =  [\hat{g}_{1k} \ \dots \ \hat{g}_{Mk}]^T \in \Cb^{M \times 1}$ contains estimated CSIs for the $k$-th UE.  

Since arbitrary phase shift will not change the norm squared value, we can always find a phase rotation $ \theta_k $ that satisfies
\begin{equation}\label{rot}
\sqrt{ \left| \hat{\gv}_k^T \tilde{\wv}_k \right|^2 }
= \sqrt{ \left| \hat{\gv}_k^T \wv_k e^{-j \theta_k }\right|^2 }
= \hat{\gv}_k^T \tilde{\wv}_k.
\end{equation}

Therefore, by properly choosing the phase of the beamforming vector, we can force $ \hat{\gv}_k^T \wv_k $ to be real positive without losing the optimality of the solution \cite{bjornson2014optimal}. Then \eqref{qos} becomes
\begin{equation} \label{realqos}
\begin{split}
\frac{1}{\sqrt{\gamma_0}} \Real \left( \hat{\gv}_k^T \wv_k \right)
&\geq \sqrt{ \sum_{i\ne k}^K \left| \hat{\gv}_k^T \wv_i \right|^2
+ \sum_{m=1}^M \sum_{i=1}^K \delta_{mk} \left|w_{mi}\right| ^2
+ \frac{1}{\rho_d} },\\
\Image \left( \hat{\gv}_k^T \wv_k \right) &= 0.
\end{split}
\end{equation}

To further simplify the notation, define $ \deltav_k = [\delta_{1k}, \dots, \delta_{Mk}]^T $, which contains the variance of channel estimation errors for the $k$-th UE. Then the second term on the right hand side (RHS) of \eqref{qos} can be written as 
\begin{equation}
\sum_{m=1}^M \sum_{i=1}^K \delta_{mk} \left|w_{mi}\right| ^2
= \sum_{m=1}^M \sum_{i=1}^K \left| \delta_{mk}^{1/2} w_{mi}\right| ^2
= 
\begin{Vmatrix}
\deltav_k^{1/2} \circ \wv_1 \\
\vdots \\
\deltav_k^{1/2} \circ \wv_K
\end{Vmatrix}_2^2,
\end{equation}
where $\av \circ \bv$ is Hadamard product between vectors $\av$ and $\bv$. 

Based on the three terms on the RHS of \eqref{qos}, define 
\begin{equation}\label{vk}
\vv_k
= \left[\vv_{\text{MUI}_k}^T \  \vv_{\text{CEE}_k}^T \ \frac{1}{\sqrt{\rho_d}} \right]^T,
\end{equation} 
where 
\begin{align}
\vv_{\text{MUI}_k} 
&= \left[\hat{\gv}_k^T \wv_1 \ \dots \hat{\gv}_k^T \wv_{k-1} \ \hat{\gv}_k^T \wv_{k+1} \ \dots \hat{\gv}_k^T \wv_K \right]^T, \label{vmui} \\ 
\vv_{\text{CEE}_k} 
&= \left[\left(\deltav_k^{1/2} \circ \wv_1 \right)^T \ \dots \ \left(\deltav_k^{1/2} \circ \wv_K \right)^T \right]^T. 
\label{vcee}
\end{align}
Then \eqref{qos} and \eqref{realqos} can be rewritten as
\begin{align} \label{soc}
\frac{1}{\sqrt{\gamma_0}} \Real \left( \hat{\gv}_k^T \wv_k \right)
\geq \| \vv_k \|, ~~
\Image \left( \hat{\gv}_k^T \wv_k \right) = 0.
\end{align}

Based on the above analysis and notations, now we are ready to prove the quasi-concave property of (P1). 

\begin{proposition}
Problem (P1) is quasi-concave.
\end{proposition}
\begin{proof}
Since the constraint in (P1) is quadratic thus convex, it is sufficient to show that the objective function of (P1) is a quasi-concave function, that is, the upper level set of the objective function is a convex set. 

Define the objective function of (P1) as 
\begin{align}
f(\Wv) = \underset{k\in \{1,\dots,K\}}{\min} 
  \gamma_k.
\end{align}
For any given $\gamma > 0$, the upper level set of $f(\Wv)$ can then be evaluated as
\begin{align*}
\Ub(\gamma) &= \left\{\Wv: f(\Wv) \geq \gamma \right\}, \\
&= \left\{\Wv: \gamma_k \geq \gamma, \text{~for~} k = 1, \cdots, K \right\}, \\
&= \left\{\Wv: \frac{1}{\sqrt{\gamma}} \Real \left( \hat{\gv}_k^T \wv_k \right)
\geq \| \vv_k \|, \text{~for~} k = 1, \cdots, K \right\},
\end{align*}
where the last equality is based on \eqref{soc}. Since $\Ub(\gamma)$ is a second-order cone, it is convex. Thus the objective function is quasi-concave.
\end{proof}

Since (P1) is quasi-concave, we can solve it by performing bisection search over its objective function. Specifically, based on the analysis in this section, (P1) can be converted to an equivalent problem as
\begin{align*}
\text{(P2)~~~} \underset{\Wv, \gamma}{\text{maximize}} ~~~~& \gamma \\ 
\text{subject to} ~~~~ &\frac{1}{\sqrt{\gamma}} \Real \left( \hat{\gv}_k^T \wv_k \right) \geq \| \vv_k \|, \ k=1,\dots, K,\\
& \Image \left( \hat{\gv}_k^T \wv_k \right) = 0, \ k=1,\dots, K, \\ 
&\sum_{k=1}^K \left|w_{mk}\right| ^2 \leq 1,
\ m=1,\dots, M.
\end{align*}
The above problem can be optimally solved by combining bisection search with a convex feasibility problem in a manner that is similar to \cite[Algorithm 2]{ngo2017cell}. It should be noted that he complexity of the bisection algorithm with convex optimization scales in polynomial time with the network size, and it can be quite high for systems with a large number of APs and/or UEs.

\section{Simulation Results}
Simulation results are provided in this section to study the performance of optimum beamforming in cell-free massive MIMO systems. The results are also compared to systems with conjugate or ZF beamforming. The simulation of CB is implemented by following \cite{interdonato2016much} and \cite{interdonato2019downlink}, and the simulation of ZF is extended from \cite{nayebi2017precoding} by including downlink training as in \cite{ngo2013massive}. 
In data transmissions, all three beamforming schemes utilize max-min power control mechanisms to ensure fairness among users. The power control for ZF and CB is performed by using the methods described in \cite{nayebi2017precoding} and \cite{ngo2017cell}, respectively. For all three beamforming methods, including OB, ZF and CB, power normalization is performed by imposing a power constraint for each AP, and the optimum power allocation among APs is performed on the network level by considering the interactions among APs. 

The simulation is performed over a square area with size $ 1 \times 1 \ \text{km}^2$ and it is wrapped around the boundary. The APs and UEs are randomly placed following uniform distribution. The propagation environment is considered to be an urban area using the Hata-COST231 model \cite{cichon1999propagation} 
. The large scale fading coefficient is modeled as
\begin{equation}\label{large-scale}
\beta_{mk}(\text{dB})
= -L-3.5\times 10\log_{10}(d_{mk})+z_{mk},
\end{equation}
where $ L=140.72 $ is the pathloss at a reference distance, the pathloss exponent is 3.5, $ d_{mk} $ is distance between $ m $-th AP and $ k $-th UE in kilometer, and $ z_{mk} $ is shadowing follows i.i.d normal distribution with zero mean and variance $ \sigma_{\text{sh}}^2 = 8 $ dB. 

The RF carrier frequency is assumed to be 1.9 GHz, and the system have bandwidth $ B = 20 $ MHz. The noise power is calculated as $ \sigma_w^2 = B \times k_B \times T_0 \times \sigma_n $, where $ k_B $ is the Boltzmann constant, $ T_0 = 290K $ is the noise temperature and $ \sigma_n = 9 $ dB is the noise figure. The actual transmission power $ \bar{\rho}_p $ and $ \bar{\rho}_d $ is 23 dBm, and the normalized SNR can be calculated as the actual power divided by the noise power $\sigma_w^2$.  

In order to account for the overhead of pilot symbols in channel estimation, define the net throughput for user 
$k$ as
\begin{equation}\label{throughput}
S_k = \frac{B}{2} \left(1- \frac{\tau_p+\tau_b}{\tau_c}\right) \log_2(1+\gamma_{\text{UE},k}),
\end{equation}
where $ \tau_c $ is coherent interval and the factor $1/2$ is used to account for the fact that half of the coherence interval are used for downlink transmission in TDD schemes. The system performance is evaluated by using $S_k$.

During simulations, all pilot sequences are assigned in a random manner. It should be noted that optimum pilot assignment can have significant impacts on systems with pilot contamination. However, only random pilot assignments are considered in this paper due to space limit.

We first compare the net throughput per user for systems with different beamforming schemes. All systems are equipped with $M = 100$ APs and $K = 40$ UEs. The coherent interval is set to 400 and the downlink pilot length is $ \tau_b = K $. Fig. \ref{S_CDF} shows the empirical cumulative distribution functions (CDFs) of the net throughput per user generated from 200 realizations with random geometric locations, shadowing, small-scale fading, and pilot sequence assignments. The proposed OB scheme achieves the best performance, followed by ZF and CB, respectively. At the $5\%$ outage rate, the net throughput of OB, ZF, and CB without pilot contamination ($\tau_p = K$) are approximately 28, 25, and 9.5 Mbps, respectively. These numbers are changed to 23, 19.5, and 9.5 Mbps when there is pilot contamination ($\tau_p = 20$). 

\begin{figure}[!t]
	\begin{center}
		\includegraphics[width=0.4\textwidth]{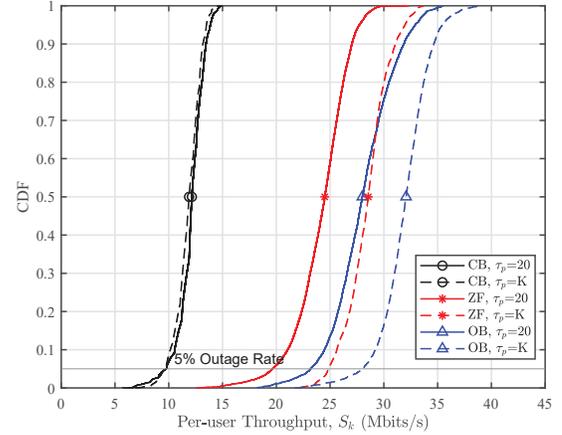}
	\end{center}
	\caption{CDFs of per-user downlink throughput with or without pilot contamination for CB, ZF, and OB.} \label{S_CDF}
\end{figure}

The performance improvement of OB is achieved at the cost of a higher computation complexity. The OB precoder needs to be optimized for each coherence interval of samll-scale fading, while the power control for ZF and CB are performed based on large-scale fading. The low net throughput of CB is mainly due to the fact that beamforming is performed locally at each AP without cooperation among APs. On the other hand, the beamforming of both OB and ZF is performed at the CU with global CSI. This cooperated design requires CSI exchange among APs and CU, which will cause significant increases in communications and overheads in front-haul links. Therefore, the performance of OB and ZF are more sensitive to non-ideal front-haul links with limited capacities.
Among the three schemes, ZF is the most sensitive to pilot contamination, while OB is slightly less affected. It is interesting to note that the net throughput of CB is slightly increased when reducing $\tau_p$ from 40 to 20. This means that the benefits of smaller overhead out-weight the negative effects of pilot contamination. Thus CB is robust against channel estimation errors.

Fig. \ref{S_PLength} studies the impacts of uplink pilot length and training overhead on the average and minimum per-user net throughput for different beamforming schemes. All systems are equipped with $100$ APs and $80$ UEs. The coherent interval is assumed to be 300. The downlink pilot sequence is $\tau_b = K = 80$. Both the average and minimum net throughput is concave in pilot length for all three beamforming schemes. The concavity indicates the tradeoff between overhead and pilot contamination. When the pilot length is small, the system performance is dominated by channel estimation errors due to the effects of pilot contamination. When the pilot length is large enough, the effects of overhead out-weights those of pilot contamination. The CB, ZF, and OB achieves the maximum net throughput at $\tau_p =$ 20, 40, and 40, respectively. CB achieves the peak net throughput with a shorter pilot, and this is consistent to the fact that CB is more robust to pilot contaminations.
\begin{figure}[!t]
	\begin{center}
		\includegraphics[width=0.4\textwidth]{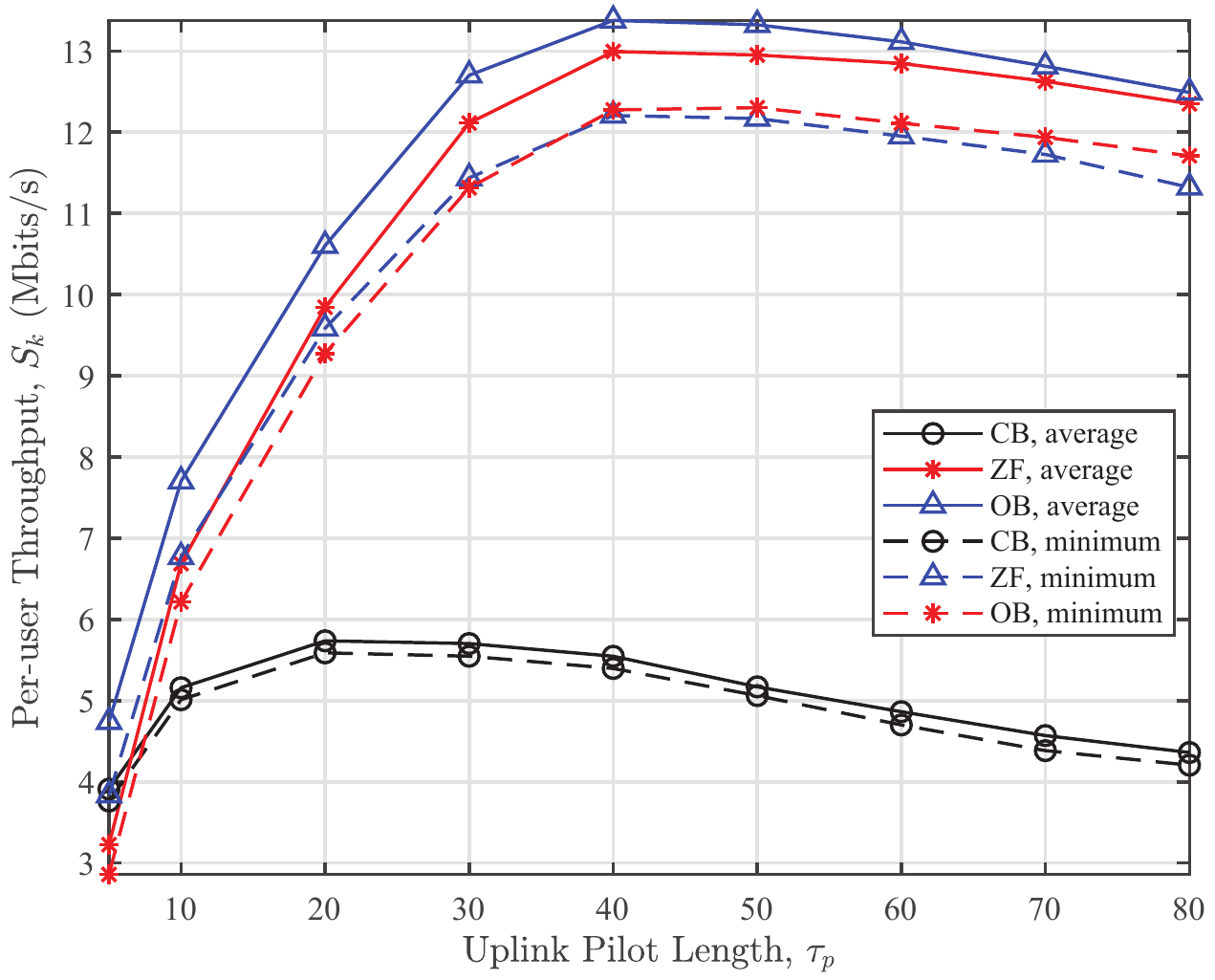}
	\end{center}
	\caption{Per-user throughput with different pilot length for CB, ZF, and OB with $ M $=100, $ K $=80.} \label{S_PLength}
\end{figure}

The complexities of the various algorithms are compared in Table \ref{time_comp} in terms of CPU time required during one iteration. {The simulations were performed on a Windows 10 workstation with 3.30 GHz i7-5820K CPU and 32.0 GB of random access memory.} The main complexities of CB and ZF come from power control, which is performed based on large scale fading. As expected, the complexity of OB is higher than both CB and ZF due to joint optimization of precoders and power control.

\begin{table}[!t]
\vspace{-0.5cm}
		\caption{CPU time (in seconds) comparison of OB, CB, ZF}\label{time_comp}
\vspace{-0.5cm}
	\begin{center}
	\begin{tabular}{|c|c|c|c|c|} 
		\hline
		 {\bf Number of users } & {\bf 10} & {\bf 20} & {\bf 30} & {\bf 40} \\ \hline
		{\bf OB} & {\color[HTML]{000000} 45.35241} & {\color[HTML]{000000} 176.6613} & {\color[HTML]{000000} 571.7387} & {\color[HTML]{000000} 1340.513} \\ \hline
		{\bf CB} & {\color[HTML]{000000} 51.90171} & {\color[HTML]{000000} 113.7298} & {\color[HTML]{000000} 217.6504} & {\color[HTML]{000000} 427.1219} \\ \hline
		{\bf ZF} & {\color[HTML]{000000} 3.599214} & {\color[HTML]{000000} 4.269728} & {\color[HTML]{000000} 5.192663} & {\color[HTML]{000000} 6.100092} \\ \hline 
	\end{tabular}
\end{center}
\vspace{-0.6cm}
\end{table}

\section{Conclusion}

The optimum beamforming of cell-free massive MIMO systems has been studied in this paper. With the knowledge of estimated channels from all APs, the optimum beamforming was performed at the CU by maximizing the minimum instantaneous SINR of all users. Under the max-min criterion, the optimum beamformer can simultaneously achieve beamforming and power control, thus provide uniformly good performance to all users in the network. Simulation results demonstrate that the proposed optimum beamforming scheme outperform zero-forcing and conjugate beamforming under all system configurations, but with a higher complexity. Among the three schemes, zero-forcing is most sensitive to pilot contamination, yet conjugate beamforming is very robust against pilot contamination.

\bibliographystyle{ieeetr}
\bibliography{References}

\end{document}